\documentclass[12pt]{article}%
\usepackage{amsfonts}
\usepackage{amsmath}%
\usepackage{amssymb}%
\usepackage{graphicx}
\usepackage{amsthm}
\usepackage{setspace}
\usepackage[margin=1in]{geometry}
\usepackage{fancyhdr}
\usepackage{tabularx}
\usepackage{multirow}
\usepackage{graphicx,psfrag,epsf,epstopdf}
\usepackage{url}
\usepackage[utf8]{inputenc} 
\usepackage{colortbl}
\usepackage{bbm}  
\usepackage{bbold} 
\usepackage{caption} 
\usepackage{caption}
\usepackage{subcaption}

\usepackage[overload]{empheq}
\usepackage{cleveref}

\usepackage{natbib}
 \newtheorem{theorem}{Theorem}

 \newtheorem{remark}{Remark}

\newcommand{\VV}[1]{\textrm{Var}(#1)}
\newcommand{\ii}{\itshape}
\newtheorem{definition}{Definition}

\DeclareSymbolFont{bbold}{U}{bbold}{m}{n} 
\DeclareSymbolFontAlphabet{\mathbbold}{bbold} 

\begin{document}

\title{A New Family of Regression Models for $[0,1]$ Outcome Data: Expanding the Palette}
\author{Eugene D. Hahn \\
Department of Information \\
and Decision Sciences\\
Salisbury University\\
Salisbury, MD 21801 USA}

\maketitle
\begin{abstract}

Beta regression is a popular methodology when the outcome variable $y$ is on the open interval $(0,1)$.   When $y$ is in the closed interval $[0,1]$, it is commonly accepted that beta regression is inapplicable. Instead, common solutions are to use augmented beta regression or censoring models or else to
subjectively transform the endpoints or rescale $y$ to allow beta regression.  We provide an attractive new approach with a family of models that treats the entirety of
$y\in[0,1]$ in a single model without rescaling or the need for the complications of augmentation or censoring. This family provides
the interpretational convenience of a single straightforward model for the expectation of $y \in [0,1]$ over its entirety.  We establish conditions for the existence of a unique MLE and then examine this new family of models from both maximum-likelihood and Bayesian perspectives.  We successfully apply the models to employment data in which augmented beta regression was difficult due to data separation.  We also apply the models to healthcare panel data that were originally examined by way of data transformation.


\end{abstract}


\section{Introduction}
Beta regression modeling for bounded continuous data $y \in (0,1)$ has attracted increasing interest over the past 20 years.  Percentage outcome data such as vulnerability to climate change \citep{tranetal22}, compliance with tax requirements \citep{dezsoetal21}, and health-related quality of life \citep{gheorgeetal17} are recent examples.  Pioneering papers on beta regression include \citet{paolino01}, \citet{kies03},and \citet{ferrari04}.  The beta distribution is appealing due to its simplicity combined with its flexibility for displaying different shapes.
However the flexibility of the beta distribution declines when $y \in [0,1]$.  As a result, alternative models are often used with augmented beta regression being a popular choice.
In this paper we present a new family of models for $y \in [0,1]$ data.  We adopt a mixture modeling approach using beta and power family distributions.  This allows us to model the entirety of $y$ simultaneously without excluding the endpoints as in augmented beta regression, yielding different and complementary research insights.

\section{Literature Review\label{sec:litrev}}
In the standard parameterization of the beta distribution, $y$ follows the PDF
\begin{equation}
p(y|\alpha,\beta)=%
\dfrac{\Gamma(\alpha+\beta)}{\Gamma(\alpha)\Gamma(\beta)}
y^{\alpha-1}(1-y)^{\beta-1},
\label{eq:StdBetaDist}
\end{equation}
where $\alpha>0$ and $\beta>0$ if $0 \leq y \leq 1$.
It is also common to transform such that
\begin{subequations}\label{eq:muphibetatogether}
\begin{align}
\mu_{\beta} & = \alpha/(\alpha+\beta), \label{eq:mubeta} \\
\phi      & = \alpha+\beta \label{eq:phibeta}
\end{align}
\end{subequations}
yielding the alternative parameterization
\begin{equation}
p(y|\mu_\beta,\phi)=%
\dfrac{\Gamma(\phi)}{\Gamma(\mu_{\beta}\, \phi)\Gamma\bigl((1-\mu_\beta)\phi\bigr)}
y^{\mu_{\beta}\, \phi-1}(1-y)^{(1-\mu_\beta)\phi-1},
\label{eq:ReparBetaDist}
\end{equation}
where $0 < \mu_\beta < 1$ and where $\phi > 0$ has an interpretation as a precision component. Next the data analyst collects an outcome variable $y$ having $n_\beta$ observations and specifies
\begin{equation}\label{eq:classicbetareg}
        g(\mu_{\beta_i}) =  \eta_{\beta_i} = b_0 + b_1 x_{1i} + \ldots + b_J x_{Ji}
\end{equation}
where $x$ refers to the predictor variables, $i$ indexes the observations, $j=1, \ldots, J$ indexes the predictor variables, and the link function $g(\cdot)$ is chosen by the data analyst. Then \eqref{eq:ReparBetaDist}  and \eqref{eq:classicbetareg} define a model which we call the classic beta regression model to ease future discussion.  The parameter $\phi$ may be taken to be a single quantity or optionally may be considered to depend on predictors analogously to \eqref{eq:classicbetareg}. Numerous authors have expanded upon classic beta regression. For example spatio-temporal beta regression models have been discussed by \citet{kaufeldetal14} and \citet{mandaletal16}. Time-series beta regression models have been discussed by \cite{rocha08}, \cite{guolo14}, and \citet{dasilva16}.  Random effects and mixed effects versions can be found in \citet{zimprich10},  \citet{hunger12}, and \citet{Verkuilen12}.  \citet{zhaoetal14} examined generalizations of the beta regression dispersion structure.  \citet{SouzaMoura16} described multivariate beta regression models while variable dispersion beta regression having parametric link functions have been discussed by \cite{canterlebayer19}.

One key limitation of the classic beta regression model is that endpoint $y$ values of zero and one traditionally have posed challenges.  Many authors have used the augmented (also called inflated or hurdle) beta regression models for such data \citep{cooketal08, galivs14, wangandluo16, wangandluo17, dibriscoandmigliorati20}.  Here the analyst partitions the original data of sample size $N$ into multiple subsets.  Denote $y_\beta$ as the subset where $y\in(0,1)$ with sample size $n_{y_\beta}$. Next if $y=0$ observations exist, a variable $z_0$ that takes the value 1 when $y=0$ and zero otherwise is created with  $n_{z_0}$ successes and $N - n_{z_0}$ failures.  If required, $z_1$ is analogously created.
The generated data $z_1$ and/or $z_0$ then augments $y_\beta$.
One assumes that the likelihood of observing $z_0$,  $z_1$ and $y_\beta$ are independent \citep{ospina10}, yielding the likelihood function
\begin{equation}
\label{eq:likeZOIB0}
\ell(y|\cdot) = \prod_{i=1}^{n_\beta} p(y_{\beta_i}|\cdot)
\prod_{i=1}^{N} p(z_{0i}|\cdot)
\prod_{i=1}^{N} p(z_{1i}|\cdot)
,
\end{equation}
where a given $z$ may be omitted and $\cdot$ is a placeholder for parameters to be estimated.  Binary choice models are used for $z_0$ and/or $z_1$.  The classic beta regression likelihood,
$\ell(y_\beta|\cdot) = \prod_{i=1}^{n_\beta} p(y_{\beta_i}|\cdot)$,
appears as the first product term in \eqref{eq:likeZOIB0}.

Since the product terms in \eqref{eq:likeZOIB0} are independent, the numerical results from a maximum likelihood augmented beta regression model can be exactly duplicated by estimating a maximum likelihood classic beta regression with $y_\beta$, and then estimating separate binary choice models on $z_0$ and $z_1$ using maximum likelihood \citep{ospina10}.
This may address important research questions.  However, observe that modeling the overall expectation of $y \in [0,1]$ by combining regression results from the tripartite vector of variables $(y_\beta, z_1, z_0)$ and correct sample sizes of $n_\beta, n_{z_0}$, and $n_{z_1}$ will not be straightforward using the results from \eqref{eq:likeZOIB0}.
It might be challenging to explain that the endpoint of some variables, such as a test score of 100\%, arises from an entirely independent statistical process from that of an adjacent value, such as a score of 99.9\%.
Another area where interpretation may be more involved are models for longitudinal data.  If at time $t$ we observe $y_{it}$, $y_{it}$ may fluctuate between product terms in \eqref{eq:likeZOIB0} and be treated as continuous or point-valued depending on $t$.
This may be theoretically awkward and may lead to a loss of precision if the random effect for observation $i$ must be estimated in multiple product terms and somehow combined.
A second solution to this problem is to consider models for censored data \citep{tobin58}.
Doubly censored models are also possible (\citealp{smithson06, ospina10}) as is a combination of censoring and augmentation \citep{hwangetal21}. Censoring models may be more interpretable for data like exam scores but the computational overhead may be nontrivial when time-series or spatiotemporal modeling is of interest.
A third, more recent, solution is to use a mixture distribution that has positive, finite support on $y \in [0,1]$.
Regression modeling using a mixture of the beta and rectangular distributions was first proposed by \citet{bayesetal12} and subsequently considered by \cite{Ribeiroetal21}.  Unfortunately a clear methodology for the endpoints was not proposed.
Instead, \citet{hahn21} proposed a regression model where $y\in[0,1]$ was taken to be a mixture of a beta distribution and a tilting distribution. However, the expectation of the tilting distribution lies in the range $[1/3, 2/3]$, partially constraining marginal inference for $y$ in those models. The current paper proposes models that remove this limitation.

\section{The Tilting Power Distribution\label{sec:tiltingpowerdist}}

Tilted beta regression makes use of the tilting distribution which is a mixture of triangular distributions.
The triangular distribution is attractive in its simplicity but this comes at a cost of flexibility.  \citet{vandorp02b} described the two-sided power distribution as a more flexible alternative to the triangular.  Our requirements are first that the distribution is to provide finite density at the endpoints and second to avoid bimodality which would be less attractive for regression in this context.  Given the above,
we begin with a distribution for $0<y<1$ with PDF
\begin{equation}
p(y|\mu_T)= f_0=%
\begin{cases}
(|\nu|+1) (1-y)^{|\nu|}  & \text{if }  \nu<0,\\
(|\nu|+1) y^{|\nu|}      & \text{if } \nu\geq 0,
\end{cases}
\label{eq:newDistNoEnds}
\end{equation}
with $\nu \in \mathbb{R}$ expressed in terms of the expectation of \eqref{eq:newDistNoEnds}, $\mu_T$, as
\begin{equation}
\nu = \dfrac{2 \mu_T-1}{1/2 -|\mu_T-1/2|}.
\label{eq:invertmuTtonu}
\end{equation}
The sign of $\nu$ (equivalently the position of $\mu_T$ versus the midpoint) governs the tilt with $\nu=0$ ($\mu_T=1/2$) giving the uniform distribution. The triangular distribution with mode of 0 arises when $\nu=-1$ while the triangular distribution with mode of 1 arises when $\nu=1$.  We can describe \eqref{eq:newDistNoEnds} as the tilting power distribution. We find
$E(Y) = \mu_T = \dfrac{\nu I(\nu>0) +1}{|\nu| + 2}$ 
where $I(\cdot)$ is the indicator function and $\text{Var}(Y)  = \dfrac{|\nu| + 1}{(|\nu|+3)(|\nu|+2)^2}$. Here $0 < \mu_T <1$ which improves on the tilting distribution of \citet{hahn21}.
As \eqref{eq:newDistNoEnds} is a special case of the two-sided power distribution, \citet{vandorp02b} give further distributional details.
Figure \ref{fig:tiltpowerdens} contains a plot of \eqref{eq:newDistNoEnds} for several values of $\nu$.

\begin{figure}[ht]
\begin{center}
\caption{Tilting Power Densities for Selected Values of $\nu$}
\label{fig:tiltpowerdens}
\includegraphics[scale=0.7]{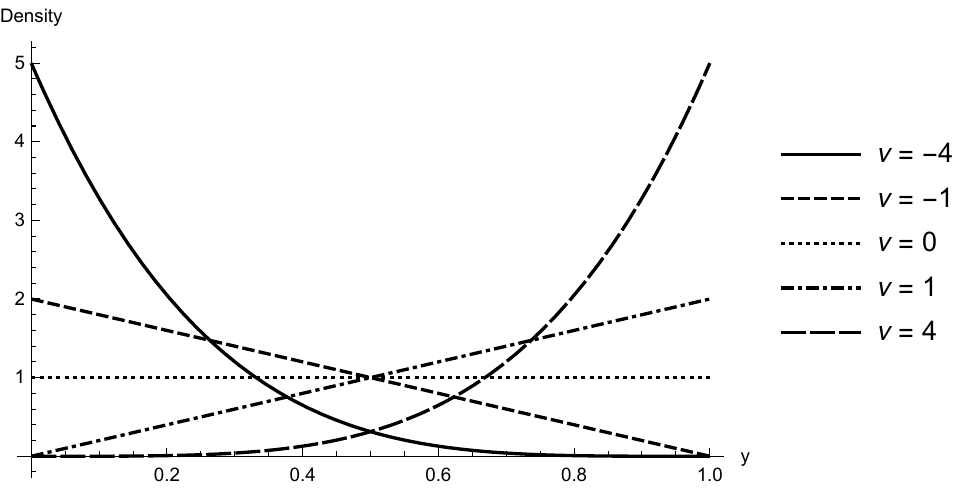}
\end{center}
\end{figure}
We can handle the endpoints by recognizing that the reflected power distribution
\begin{equation}
p(y|\mu_T)=%
\frac{1-\mu_T}{\mu_T}(1-y)^{\frac{1-\mu_T}{\mu_T}-1} 
\label{eq:reflectedpowerdistcleaned}
\end{equation}
always has finite, non-negative likelihood for all $\mu_T$ when $y=0$.  We thus require any observations where $y=0$ to arise from \eqref{eq:reflectedpowerdistcleaned}. Next consider the power distribution
\begin{equation}
p(y|\mu_T)=%
\frac{\mu_T}{1-\mu_T} y^{\frac{\mu_T}{1-\mu_T}-1}
\label{eq:powerdistcleaned}
\end{equation}
and similarly require that observations where $y=1$ arise from \eqref{eq:powerdistcleaned}.  This gives
\begin{equation}
 f_1(\mu_T)=%
\begin{cases}
\dfrac{1-\mu_T}{\mu_T}         & \text{if }  y=0,\\
f_0                            & \text{if } 0< y< 1,\\
\dfrac{\mu_T}{1-\mu_T}         & \text{if }  y=0.
\end{cases}
\label{eq:newDistEnds}
\end{equation}
We now designate $f_1$ and $f_2$ as the two components of the mixture, giving
\begin{equation}
p(y|\alpha, \beta,\mu_T, \theta)=%
(1-\theta) f_1(\mu_T) +
\theta    \dfrac{\Gamma(\alpha+\beta)}{\Gamma(\alpha)\Gamma(\beta)}
y^{\alpha-1+I(y=0)}(1-y)^{\beta-1+I(y=1)}
\label{eq:tiltingpowerbetadist}
\end{equation}
where $0<\theta<1$ is the mixing weight and $f_2$ is the beta component.
The indicator functions in $f_2$ ensure $f_2$ is zero at an endpoint which is viable here due to the presence of $f_1$ with positive density.   Then for \eqref{eq:tiltingpowerbetadist}
\begin{align}
E(Y) &= (1-\theta)\mu_T + \theta \mu_\beta = (1-\theta)\dfrac{\nu I(\nu>0) +1}{|\nu| + 2} + \theta \dfrac{\alpha}{\alpha+\beta},\label{eq:expectationofexpectations1}
\end{align}
and we may write a likelihood function as
\begin{equation}
\label{eq:liketiltingpower}
\ell(y|\alpha,\beta, \mu_T, \theta) = \prod_{i=1}^{N} p(y_{i}|\alpha,\beta, \mu_T, \theta) = \prod_{i=1}^{N}(1-\theta) f_1(y_i|\mu_T) + \theta f_2(y_i|\alpha, \beta).
\end{equation}
To summarize, a key difference is that \eqref{eq:newDistEnds} permits $0 <\mu_T < 1$ which allows for improved fit in modeling versus previous work.  Additional differences are discussed in Section \ref{sec:conc}.

\subsection{Predictive Modeling}

\subsubsection{Model 1 - tilting power beta regression}
We may perform predictive modeling by specifying
\begin{subequations}\label{eq:model1}
\begin{align}
        g_\beta(\mu_{\beta_i}) &= \eta_{\beta_i} = b_0 + b_1 x_{1i} + \ldots + b_J x_{Ji}, \label{eq:mubetamodel1} \\
        g_T(\mu_{T_i})   &= \eta_{T_i} = c_0 + c_1 x_{1i} + \ldots + c_J x_{Ji}, \label{eq:muTmodel1} \\
        g_\phi(\phi_i) &= \eta_{\phi_i}= d_0 + d_1 x_{1i} + \ldots + d_J x_{Ji}, \label{eq:phimodel1} \\
        g_\theta(\theta_i)  &= \eta_{\theta_i}= a_0 + a_1 x_{1i} + \ldots + a_J x_{Ji}. \label{eq:thetamodel1}
\end{align}
\end{subequations}
We drop the subscript for $g$ going forward. Here \eqref{eq:tiltingpowerbetadist}, \eqref{eq:liketiltingpower} and \eqref{eq:model1} resemble those of ordinary finite mixture models \citep[e.g.,][\S1.4.4]{titterington11} upon quick review. However $f_2$ contributes only $n_\beta$ observations to our understanding of $y$ since the indicator functions remove the $N-n_\beta$ endpoints from $f_2$.  Thus we have two latent classes for $y_\beta$ but only one (non-)latent class for the endpoints of $y$.

Suppose we have no interest in possible latent classes and are mainly interested in $y$.  Equating $\mu_\beta$ and $\mu_T$ in \eqref{eq:expectationofexpectations1} provides a useful $\mu_\beta = \mu_T = E(y)$. For \eqref{eq:muTmodel1}, we replace $\mu_T$ with $\mu_\beta$ using \eqref{eq:invertmuTtonu} to obtain the contribution previously provided.  We then continue with estimation of \eqref{eq:liketiltingpower}, \eqref{eq:mubetamodel1}, \eqref{eq:phimodel1}, and \eqref{eq:thetamodel1}.  This special case is called Model 1.  Estimation of \eqref{eq:mubetamodel1} has a useful interpretation as the conditional expectation of $y \in [0,1]$.  Contrast this with augmented beta regression where the endpoints are excluded and probability models for endpoints are used instead.  Label-switching is of no particular importance to real-world interpretation of Model 1 since the means have been equated.

\subsubsection{Model 2}
In Model 1 we allow $\theta$ to be estimated under the presumption of unknown latent classes.  Endpoints contribute to the likelihood through $f_1$ in \eqref{eq:liketiltingpower} while $y_\beta$ contributes through $f_1 $and $f_2$.  In Model 2 we now assign the endpoints to $f_1$, assign $y_\beta$ to $f_2$ only, and then estimate weights $(1-\theta)$ and $\theta$ respectively per \eqref{eq:liketiltingpower}.
We can call Model 2 a semimixture model since the mixture weights are retained but the traditional mixture specification of latent classes and component-specific means has been discarded.

\subsubsection{Model 3 - the endpoint heterogeneous beta regression model}\label{sec:ehbetatheory}

We now further simplify Model 2 by eliminating $\theta$ and \eqref{eq:thetamodel1} while retaining Model 2's $\mu_\beta = \mu_T = E(y)$.  The power function distribution is equivalent to \eqref{eq:StdBetaDist} with $\beta=1$ and it has finite non-zero density at $y=1$. Also the reflected power distribution is equivalent to \eqref{eq:StdBetaDist} with $\alpha=1$ and has finite non-zero density at $y=0$.  Thus we directly apply the relevant distribution to endpoint observations.   Additional details are as follows.
We require that $\mu_\beta$ remains free to vary in \eqref{eq:ReparBetaDist}, hence $\phi$ must be fixed.   Application of \eqref{eq:mubeta} and \eqref{eq:phibeta} yields $\alpha = \mu_{\beta}\, \phi$ and $\beta= (1-\mu_{\beta})\phi$ for insertion into \eqref{eq:StdBetaDist}.
Therefore when $y=1$, applying the constraint $\beta=1$ gives
\begin{equation}\label{eq:mualphapowerdist}
  \alpha=  \frac{\mu_T}{1-\mu_T}
\end{equation}
and in the alternate parameterization we set
\begin{equation}\label{eq:phiphipowerdist}
  \phi=  \frac{1}{1-\mu_T}
\end{equation}
where $\mu_T$ has been used for continuity purposes in discussion of the former distribution component $f_1$ ($\mu_\beta$ could have been used equivalently). When $y=0$, applying the constraint $\alpha=1$ gives
\begin{equation}\label{eq:mubetareflectedpowerdist}
  \beta=  \frac{1-\mu_T}{\mu_T}
\end{equation}
or alternatively
\begin{equation}\label{eq:phiphireflectedpowerdist}
  \phi=  \frac{1}{\mu_T}.
\end{equation}
As in Model 2, estimation of \eqref{eq:muTmodel1} is subsumed into estimation of \eqref{eq:mubetamodel1} and is not performed.  Next we prove a theorem regarding Model 3.

\begin{theorem} \label{thm:classicbetaregression}
  Model 3 constitutes a beta regression model for $y\in[0,1].$
\end{theorem}
\begin{proof}
  Here we use $\mu$ in place of $\mu_\beta$ since both mixture components have the same mean.  Considering the above and \eqref{eq:ReparBetaDist} we can now write a beta distribution as
\begin{equation}
p(y|\mu,\phi^*)=%
\dfrac{\Gamma(\phi^*)}{\Gamma(\mu\, \phi^*)\Gamma\bigl((1-\mu)\phi^*\bigr)}
y^{\mu\, \phi^*-1}(1-y)^{(1-\mu)\phi^*-1}
\label{eq:EHBetaDist}
\end{equation}
where  $\phi^* > 0 \ni \phi^* = \phi + I(y=0)\biggl(\dfrac{1}{\mu}-\phi\biggr) + I(y=1)\biggl(\dfrac{1}{1-\mu}-\phi\biggr)$.  When $y \in (0,1)$ the classic beta regression model applies.  When $y=0$, $\phi^*$ produces the Beta$\bigl(1,(1-\mu)/\mu\bigr)$ distribution in \eqref{eq:StdBetaDist}.  When $y=1$, $\phi^*$ produces the Beta$\bigl(\mu/(1-\mu),1\bigr)$ distribution in \eqref{eq:StdBetaDist}.  Hence \eqref{eq:EHBetaDist} and \eqref{eq:mubetamodel1} constitute a beta regression model for $y\in[0,1]$.
\end{proof}
We can call Model 3 the endpoint heterogeneous beta regression model to remind us that $\phi$ is adjusted at the endpoints to produce $\phi^*$.  Endpoint observations do not provide information about $\phi$ in Model 3, which is also true of augmented beta regression.  Therefore \eqref{eq:phimodel1} is omitted for endpoint values. We can also backpropagate this result to Model 2.  Namely we replace $f_1(y|\mu_T)$ of \eqref{eq:liketiltingpower} with the alternate Beta$\bigl(\mu, \phi^*\bigr)$ distribution.  While Model 3 expands the palette of models for $y\in[0,1]$ data, some estimation challenges may arise.  We now examine this issue. For reference, \eqref{eq:mualphapowerdist} in \eqref{eq:EHBetaDist} implies the power distribution of \eqref{eq:powerdistcleaned} while the reflected power distribution from \eqref{eq:mubetareflectedpowerdist} in \eqref{eq:EHBetaDist} implies \eqref{eq:reflectedpowerdistcleaned}.
Next let $n_0\equiv n_{z_0}$ be the number of observations where $y=0$ and denote this data as $y_0$.  Also let $n_1\equiv n_{z_1}$ be the number of observations where $y=1$ and denote this data as $y_1$.
Then we may write the three separate portions of the log-likelihood function of Model 3 as
 \begin{align}
    \log \ell(y_0|\mu) & =  n_0 \bigl( \log(1-\mu) - \log  (\mu ) \bigr), \label{eq:LLy0} \\
    \log \ell(y_1|\mu) & =  n_1 \bigl(\log (\mu ) - \log \left(1-\mu \right)\bigr),\label{eq:LLy1}\\
    \begin{split}\label{eq:LLybeta}
        \log \ell(y_\beta|\mu,\phi) & =  n_\beta \bigl(\log \Gamma (\phi ) - \log \Gamma (\mu  \phi ) - \log \Gamma (\phi -\mu  \phi )\bigr) \\
            & \quad  +  ( \mu  \phi -1) \sum _{i=1}^{n_\beta} \log \left(y_i\right)
            + ( \phi-\mu  \phi -1)\sum _{i=1}^{n_\beta} \log \left(1-y_i\right).
    \end{split}
\end{align}
Previous authors have examined the behavior of \eqref{eq:LLybeta} \citep[e.g.,][]{ferrari04,kies03}. Therefore we focus on \eqref{eq:LLy0} and \eqref{eq:LLy1}. Observe $y$ appears in \eqref{eq:LLy0} and \eqref{eq:LLy1} only through $n_0$ and $n_1$ respectively (recall \eqref{eq:reflectedpowerdistcleaned} and \eqref{eq:powerdistcleaned}).
Briefly consider $y \notin (0,1)$ despite Model 3 not being intended for such.  First, if we have $y_0$ only (or $y_1$ only) clearly the maximum likelihood will be undefined (we would expect poor behavior for degenerate $y$).  Second, if we have both $y_0$ and $y_1$ where $n_0=n_1$, the likelihood is zero for $0<\mu<1$ since \eqref{eq:LLy0} and \eqref{eq:LLy1} exactly cancel. Again, estimation of Model 3 is not possible. Third, if we have $y_0$ and $y_1$ where $n_0 \neq n_1$, clearly the first $\min[n_0, n_1]$ terms will cancel and the likelihood of the remaining $|n_0-n_1|$ observations will again have an undefined maximum. The cancellation property will be useful below. Furthermore
 \begin{align}
    \frac{d}{d\mu} \log \ell(y_0|\mu) & =  \frac{n_0}{\mu^2-\mu}, \label{eq:scoreLLy0} \\
    \frac{d}{d\mu} \log \ell(y_1|\mu) & = \frac{n_1}{\mu-\mu^2   } \label{eq:scoreLLy1} \\
    \frac{d^2}{d\mu^2} \log \ell(y_0|\mu)& =  n_0 \left(\frac{1}{\mu ^2}  - \frac{1}{(1-\mu )^2}\right), \label{eq:hessLLy0} \\
    \frac{d^2}{d\mu^2} \log \ell(y_1|\mu)&= n_1 \left(\frac{1}{(1-\mu )^2} - \frac{1}{\mu ^2}\right). \label{eq:hessLLy1}
 \end{align}
Challenges include \eqref{eq:scoreLLy0} and \eqref{eq:scoreLLy1} lacking roots and \eqref{eq:hessLLy0} and \eqref{eq:hessLLy1} possibly having the wrong signs depending on the location of $\mu$ versus the midpoint of the range of $y$. However, \eqref{eq:hessLLy0} is negative when $\mu>1/2$. Combining with \eqref{eq:LLybeta} we clearly have a global maximum in the parameter space when $y\in[0,1)$ and $\mu>1/2$ given that \eqref{eq:LLybeta} has its own maximum in the parameter space. Similarly we have a global maximum in the parameter space when $y\in(0,1]$ and $\mu<1/2$ given that \eqref{eq:LLybeta} has its own maximum in the parameter space.  If these conditions do not hold, we may still find we have a local maximum in the vicinity of the maximum for \eqref{eq:LLybeta}.  
The extent to which there is a local maximum will depend on the magnitudes of $|n_0-n_1|, n_\beta, \mu$ and $\phi$.
We now examine the data conditions mentioned in Theorem \ref{thm:classicbetaregression} with respect to the existence of a global maximum.
\begin{definition}
We say the data is $J$-consistent when for a given set of data $\mu_\beta \gtrless 1/2$ and $n_1 \gtreqless n_0$.
\end{definition}
\begin{theorem}
  \label{thm:globalmax}
  Suppose we have $n_0$ observations of $y_0$ and $n_\beta$ observations of $y_\beta$ and that the classic beta regression MLEs exist for $y_\beta$.  Further assume we have $\mu_\beta < 1/2$ and that $\phi$ is large for $y_\beta$.  Then the MLEs of $\mu$ and $\phi$ for Model 3 will have a global optimum in the parameter space as long as $n_\beta \geq N/2$ and the data is $J$-consistent.
\end{theorem}
\begin{proof}
See Appendix.
\end{proof}
\begin{remark}[Absence of $J$-consistency]
Suppose the data is not $J$-consistent, i.e., suppose that we have $n_0=n_\beta$ and $n_1=0$ but that $\mu_\beta > 1/2.$   Then as $\mu_\beta$ becomes increasingly larger than 1/2, it is increasingly likely that $\VV{y} > 1/12$.  Note that when $\VV{y} = 1/12$, the mixture has the same variance as the standard uniform distribution.  When $\VV{y} > 1/12$, then the $U$-shaped beta distributions will increasingly be supported and $\phi$ is likely to go below 1.  A value of $\phi$ of 1 or lower indicates a finite MLE will no longer be guaranteed.  Such data will tend to be more consistent with the augmented beta regression model since the augmented beta regression model specifies that the endpoint is treated separately from $y_\beta$ and $\phi$.
\end{remark}
We now extend Theorem \ref{thm:globalmax} as follows.
\begin{theorem}
  \label{thm:globalmaxally}
  Suppose we have $n_0$ observations of $y_0$, $n_1$ observations of $y_1$, and $n_\beta$ observations of $y_\beta$.  Further assume $\phi>1$ in $y_\beta$ and that $J$-consistency holds.  Then the MLEs of $\mu$ and $\phi$ for Model 3 will have a global optimum in the parameter space as long as $n_\beta \geq N/2$, i.e., as long as the non-endpoint observations constitute at least half of the total observations.
\end{theorem}
\begin{proof}
  The case where $n_0>0$ and $n_1=0$ was proved above.
  Suppose $n_0=0$ and $n_1>0$.  The proof of Theorem \ref{thm:globalmax}  was based on \eqref{eq:reflectedpowerdistcleaned}.  Here we instead have \eqref{eq:powerdistcleaned}.  Clearly then the global optimum exists in the parameter space when $n_1=n_\beta$ with the associated reflection.
  Finally when both $n_0>0$ and $n_1>0$, we have shown above that only $|n_0-n_1|$ observations will contribute to the likelihood and the remaining endpoint observations will cancel.  Therefore, in this case a global maximum exists as long as $n_\beta \geq |n_0-n_1|$, completing the proof.  In this case, $n_\beta$ may actually be substantially less than $N/2$ and still satisfy $n_\beta \geq |n_0-n_1|$.
\end{proof}

Theorem \ref{thm:globalmaxally} shows that the likelihood of Model 3 has a global maximum under the above conditions.   In practice numerical investigations reveal the bound of Theorem \ref{thm:globalmaxally} can be loosened slightly. As a simple example, suppose we take $n_0 = n_\beta + \delta$ where $\delta$ is a small constant.  Then as $N \rightarrow \infty$, note that $n_\beta +\delta \approx n_\beta \approx N/2$ so estimation is likely well-behaved for small excursions relative to $N$.

In regression models $\mu$ will no longer be a scalar quantity.  A theorem is elusive due to complexity but suppose we have $n_0$ observations of $y_0$, and $n_1$ observations of $y_1$.  Then the contribution of these observations to the likelihood when $\mu$ is non-scalar is
  \begin{equation}
  \sum_{i=1}^{n_0} \log \left(1-\mu_i \right)
-\sum_{i=1}^{n_1}  \log \left(1-\mu_i \right)
+\sum_{i=1}^{n_1}  \log \left(\mu_i \right)
-\sum_{i=1}^{n_0}  \log \left(\mu_i\right).
  \end{equation}
Now we can rewrite \eqref{eq:normconsts0} as used for Theorem \ref{thm:globalmax} for non-scalar $\mu$.
However we then run into difficulty since determining whether a global optimum exists for a regression parameter $b$ also depends on the various properties of $x, y, n_0, n_1, n_\beta$ and $g(\cdot)$.  Thus we may wish to ensure $n_\beta \gg |n_0-n_1|$ for such $b$.  As a practical example, a random slopes model may be more difficult to estimate than a random intercepts model given the difficulty of ensuring $b$ has an MLE.

We see that Theorem \ref{thm:globalmaxally} applies on a per-parameter basis so some parameters may have a global optimum but not others.  A local maximum may be available in the vicinity of the classic beta regression's global maximum as long as $n_\beta$ for that parameter is not too much smaller than $N/2$ for that parameter and $J$-consistency holds. Thus a prudent estimation procedure for Model 3 would be to first estimate the classic beta regression using only $y_\beta$.  The parameter estimates from this initial fit would be stored and then used as initial values for estimation of Model 3 on $y$.

\subsubsection{Model 4}\label{sec:ehbeta4}
We have seen the likelihood may become unbounded for a particular parameter when its observations have predominantly endpoint values.  One approach would be to consider this parameter inestimable and omit it from the model.  We now examine an alternative approach to mitigate this issue.  Consider the $y=1$ endpoint and recall in \eqref{eq:phiphipowerdist} we have fixed $\phi$ to satisfy  $(1-\mu)\phi =1$.  Consider now fixing $\mu$ instead.  In doing so we would set
\begin{equation}
\label{eq:phiphipowerdistCONVERSE}
\mu = \frac{\phi-1}{\phi}
\end{equation}
which holds under the previous requirement that $\phi > 1$.  Now that $\phi$ is free to vary, \eqref{eq:phiphipowerdist} shows that $\mu$ will now approach 1 as $\phi \rightarrow \infty$ (the expectation of the transformed version of \eqref{eq:powerdistcleaned} behaves similarly).  Next let us reexamine the power distribution of \eqref{eq:powerdistcleaned}.  Making the substitution of \eqref{eq:phiphipowerdistCONVERSE} yields a likelihood of $\phi-1$ when $y=1$. We could therefore consider an upper bound for the likelihood of the power distribution to be $\phi-1$ by the above logic considering the role of $\phi$.  The appropriate logic applied to the reflected power distribution of \eqref{eq:reflectedpowerdistcleaned} also gives an upper bound of $\phi-1$ for the likelihood when $y=0$.

We next note an influential concept discussed by \citet[][p.~87]{Verkuilen12}. \cite{Verkuilen12} advised the use of data transformations for endpoints.  One option was to nudge endpoint values away from the endpoint using a personally-chosen delta such as 0.001.   This cause observations to range between, say, 0.001 to 0.999, permitting classic beta regression modeling.  Applying this concept, suppose we consider the limit of \eqref{eq:ReparBetaDist} as both $y$ and $\mu_\beta$ approach the endpoint of 1 at the same rate.  That is, define $\tilde{\mu} \equiv \mu_\beta=y$ and note that if we could vary both $y$ and $\mu_\beta$ infinitesimally away from the endpoint at the same rate, we would have
\begin{equation}
\lim_{\tilde{\mu} \to 1} p(y|\mu_\beta,\phi)= \lim_{\tilde{\mu} \to 1} p(\tilde{\mu}|\tilde{\mu},\phi)=\phi.
\label{eq:limittildemu}
\end{equation}
This would be the upper bound of the likelihood under these conditions and is close to the upper bound $\phi-1$ previously obtained.  While \citep{Verkuilen12} advise sensitivity analyses for the personal choice of delta,  here we do not select a particular numeric delta but instead consider a true infinitesimal in terms of a limit.

Now suppose we have a parameter for which Theorem \ref{thm:globalmaxally} does not hold, $\check{\mu}$.  Writing $\check{\mu} = g^{-1}(\check{\eta})$, for Model 4 we require that the likelihood satisfies
\begin{equation}
    \label{eq:boundy1}
    \dfrac{g^{-1}(\check{\eta})}{1-g^{-1}(\check{\eta}) }\leq U
\end{equation}
using the power distribution and similarly we require that
\begin{equation}
    \label{eq:boundy0}
    \dfrac{1-g^{-1}(\check{\eta})}{g^{-1}(\check{\eta})}\leq U
\end{equation}
using the reflected power distribution, where $U$ is an upper bound such as $(\phi-1)$ or $\phi$.  We now have the endpoints' contribution to the likelihood for Model 4.  This is
 \begin{equation}\label{eq:LLModel4}
    \begin{split}%
    \log \ell(y_i| \cdot) & =  I(y_i=0)I(q_i=0)\bigl( \log(1-\mu_i) - \log  (\mu_i ) \bigr) \\
    & \quad  +  I(y_i=1)I(q_i=0)\bigl( \log  (\mu_i )- \log(1-\mu_i) \bigr) + I(0<y_i<1) \log \ell(y_{\beta_i}|\cdot)\\
    & \quad  + I(y_i=0)I(q_i=1)\min\bigl(\log(1-\dfrac{g^{-1}(\eta_i)}{1-g^{-1}(\eta_i) })
        - \log  (\dfrac{g^{-1}(\eta_i)}{1-g^{-1}(\eta_i) } ) , U_i \bigr)\\
    & \quad  + I(y_i=1)I(q_i=1)\min\bigl(\log(\dfrac{g^{-1}(\eta_i)}{1-g^{-1}(\eta_i) })
        - \log  (1- \dfrac{g^{-1}(\eta_i)}{1-g^{-1}(\eta_i) } ) , U_i \bigr),
    \end{split}
\end{equation}
where $q_i$ takes the value one if Theorem \ref{thm:globalmaxally} does not hold for the parameter and zero otherwise, $U_i$ is a likelihood upper bound for a group mean such as $(\phi-1)$ or $\phi$, and  $\log \ell(y_{\beta_i})$ is the classic beta regression log-likelihood of \eqref{eq:LLybeta}. It can be shown that \eqref{eq:LLModel4} simplifies when $g$ is the logit link. In this case we can then omit $q_i$ and instead place a constraint directly on the parameter in $\check{\eta}$ for which Theorem \ref{thm:globalmaxally} does not hold.  Namely we can require that $|\check{\eta}_i| \leq \log U_i$ when we are estimating $\check{\mu}$. Then the two final terms in \eqref{eq:LLModel4} reduce to $I(y_i=0)\bigl( \log(1-\check{\mu}_i) - \log (\check{\mu}_i ) \bigr) + I(y_i=1)\bigl( \log  (\check{\mu}_i )- \log(1-\check{\mu}_i) \bigr)$. But then application of the constraint $|\eta_i| \leq \log U_i \forall \check{\mu}_i$ implies we no longer need the last two terms of \eqref{eq:LLModel4} because $\check{\mu}$ no longer needs a separate treatment in the log-likelihood function.   Instead we can use the log-likelihood from Model 3.  This is convenient because we can quickly identify situations where Theorem \ref{thm:globalmaxally} may not hold using descriptive statistics. We can then constrain the relevant parameters and use the log-likelihood function from Model 3, namely the sum of \eqref{eq:LLy0}, \eqref{eq:LLy1} and \eqref{eq:LLybeta}. It is rare to see classic beta regression with non-logit link functions, so this simplification can be used often.  However, if $\check{\eta}$ does not arise from a group mean parameter, that is to say if $\check{\eta}$ results from a continuous predictor such as $\check{\eta}_i=\check{b} x_i$ where $x_i$ is not constant across observations, then the simplification based on constraining $\check{b}$ also does not hold.  In this situation \eqref{eq:LLModel4} would again be needed.

Model 4 has some attractive features for analysts who would otherwise be rescaling.   First, there is no need to personally devise multiple delta values and perform sensitivity analyses on them.  In big-data contexts we may find that different delta choices have different kinds of impacts, complicating matters.  Next it is attractive we can constrain some (possibly few) parameters instead of rescaling a larger number of data points. Further the likelihood function can by itself move toward a larger value of $\phi$ if it improves overall model fit without the need of personal analyst choices.  Finally if the corresponding classic beta regression model has appreciable information about $\phi$, $\phi$ can provide anchoring to the likelihood for $\check{\mu}$.  The disadvantage can be seen with reference to \eqref{eq:expectationofexpectations1}.  Recall that a novel aspect of Models 2 and 3 is that they model the expectation of $Y$ over its entire range.  In Model 4 we can only make an approximate statement about the expected value of the endpoint observations (i.e, about $\mu_T$ in \eqref{eq:expectationofexpectations1}).  As a result the equality in \eqref{eq:expectationofexpectations1} must be replaced such that $E(Y)  \approx (1-\theta)\mu_T + \theta \mu_\beta$,
if the Model 4 upper bound is applied to Model 2 (the inequality also carries over to Model 3).  However, it should be noted that both augmented beta regression and classic beta regression using rescaling will also not satisfy the equality given by  \eqref{eq:expectationofexpectations1} when $y \in [0,1]$.  Predicted values from Model 4 may be created as in for Model 3 with the extra step of accounting for $U$ as needed.

\subsection{Simulation}
We examined Model 1 and 3's properties with simulation (Model 2's $b$ and $\phi$ are the same as Model 3's). We drew 50 random samples from the $\textrm{Beta}(\alpha=8,\beta=3)$ distribution to create $y_\beta$.  We estimated $\mu$ and $\phi$ for $y_\beta$ with classic beta regression.  We then added one $y=1$ observation to $y_\beta$ and re-estimated these parameters using Model 3. This process was repeated, adding one $y=1$ endpoint observation at a time, until 50 endpoints had been added ($n_1=50$).  We then removed all of the $y=1$ endpoints and then estimated Model 1 and 3 adding one $y=0$ endpoint at a time as above.  Due to eventual absence of $J$-consistency, we stopped at $n_0=49$ observations.

We present Model 3 results first (see Figure \ref{fig:simulation}).  The $x$-axis midpoint corresponds to $n_0=n_1=0$. Estimates of $\mu$ are plotted with a solid line and these can be compared to the arithmetic mean of $y$ (dashed line). For emphasis we have plotted a black circle on the classic beta regression estimate of $\mu$. We have also plotted estimates of $\log (\phi)$ with a dotted line using the secondary axis at right. On the righthand side of the Figure as $n_1$ increases, the estimate of $\mu$ from Model 3 is slightly lower than the arithmetic mean.  The Fisher information for the (reflected) power distribution, $\left(\mu^2(1-\mu)^2\right)^{-1}$, declines for values of $\mu$ farther from the endpoints whereas the Beta distribution can provide more support for such values in this dataset.  Hence the estimate of $\mu$ initially moves somewhat slowly.  However we see $\mu$ begins to approach the arithmetic mean more rapidly toward the endpoint.  
We also note $\phi$ increases slightly until $n_1=15$, then it begins to decline slightly. The results may call to mind those of contaminated two-component mixture models for outliers \citep[][\S3]{aitkinwilson80}.

\begin{figure}[ht]
\begin{center}
\caption{Model 3 Simulation Results for Different $n_0$ and $n_1$}
\label{fig:simulation}
\includegraphics[scale=0.9]{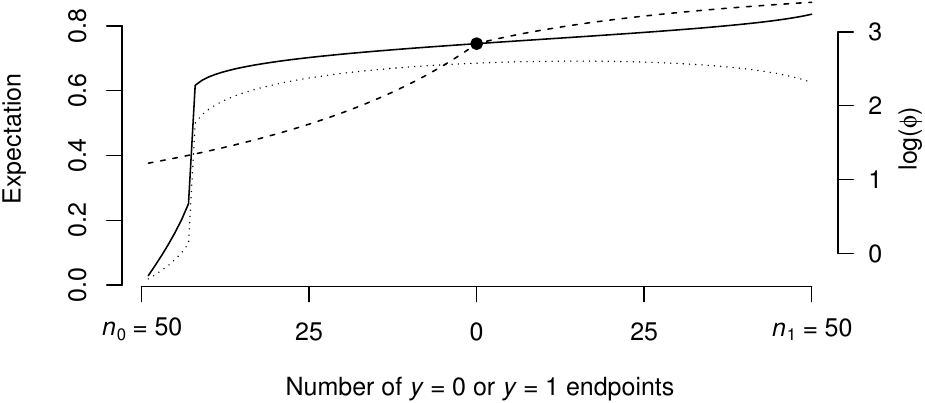}
            {\par \scriptsize{Solid line: Model 3 estimate of $\mu$, dotted line: $\log(\phi)$, dashed line: arithmetic mean of $y$.}}
\end{center}
\end{figure}

The lefthand side of the Figure presents results for data constructed to lack $J$-consistency.  Here Model 3's $\mu$ is less consistent with the arithmetic mean since the endpoint of zero is inconsistent with the bulk of the data.  The arithmetic mean falls quickly with increasing $n_0$ while $\mu$ declines relatively linearly until $n_0=42$.  Then $\log(\phi)$ falls rapidly and goes below zero at $n_0=45$.  From there the $U$-shaped beta distributions are favored and it is increasingly likely the distribution will be bimodal for greater $n_0$.

Figure \ref{fig:simulation} suggests the following practical implications.  First note that the black circle is also the estimate of $\mu$ resulting from augmented beta regression's beta component.  Its insensitivity to changes in $n_1$ and/or $n_0$ illustrates an interpretational challenge of augmented beta regression regarding $E(Y)$ when either $n_0>$ or $n_1>0$.  For very small $n_0$ and $n_1$, Model 3 is likely to lead to similar results as would the practice of rescaling.  For $J$-consistent data, rescaling a large number of values will likely inflate the value of $\phi$, possibly influencing $\phi$ to depend on the analyst's personally-chosen delta.  When the data is not $J$-consistent, rescaling may have a more pronounced and less predictable effect on $\phi$, leading to even greater dependence on the analyst's personal choices.  
In the real-world case where both $n_0>0$ and $n_1>0$ Model 3 may be attractive due to its canceling effects of $n_0$ and $n_1$, the absence of complexity that would be associated with the sensitivity analyses for nudging, and the avoidance of the complexity of augmented beta regression with two endpoints.

\begin{figure}[ht]
\begin{center}
\caption{Model 1 Simulation Results for Different $n_0$ and $n_1$}
\label{fig:simulationM1}
\includegraphics[scale=0.9]{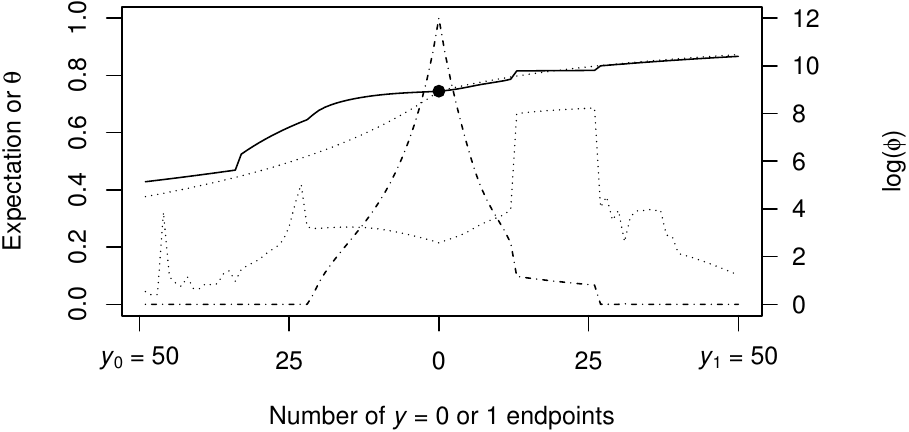}
            {\par \scriptsize{Solid line: Model 1 estimate of $\mu$, dotted line: $\log(\phi)$, dash-dotted line: $\theta$, dashed line: arithmetic mean of $y$. }}
\end{center}
\end{figure}

Figure \ref{fig:simulationM1} displays the simulation results for Model 1.  The new dash-dotted line displays the estimates for $\theta$ as $n_0$ and $n_1$ are varied.  Both $\theta$ and the expectations are on the same scale.  On the right side of the Figure, we see  Model 1 provides estimates of $\mu$ that are very close to the arithmetic mean for $J$-consistent data.  We also see that $\log(\phi)$ jumps when $13 \leq n_1 \leq 26$.  In general the choice of initial values is more important for Model 1 due to the mixture.  On the left side we see that Model 1's $\mu$ tracks the arithmetic mean better than did Model 3.  We also see this improved performance is a result of $\theta$ declining so that the beta component is effectively absent when $n_0 > 20$.  While the performance is improved versus Model 3, we see it is still somewhat difficult to fit a substantially $J$-inconsistent data set.  Augmented beta regression may thus be preferred for strongly $J$-inconsistent data if theory supports treating endpoints as conceptually distinct from non-endpoints.

\section{Applications \label{sec:empirical}}

The Financial Times \citep{FT22} provides annual data on top-100 ranked Master's of Business Adminstration (MBA) programs worldwide.  Our $y$ is the percentage of MBA graduates from a program who are employed at three months after graduation.  One to five observations of $y=100\%$ occurred per year (but no zeros) in the period 2012 to 2021.  The list composition changed from year to year as rankings changed, and occasional ties meant more than 100 programs could appear.  Stata software version 16 with package \texttt{zoib} \citep{buis12} was used for the augmented beta regression.    We had planned to use five independent variables for analysis of an arbitrary year.  However we found that the logit portion of the augmented beta regression model was poorly-behaved except for 2012.  Stata's \texttt{zoib} command produced $d$ estimates with implausibly large values (to be discussed in Table \ref{tbl:convergence}) while Stata's \texttt{logit} command provided a warning message (and sometimes would not produce estimates). Recall \citet{albertanderson84} showed that the logit model may have no MLE if the data exhibits complete or quasicomplete separation. This was the problem here. We report on a smaller more well-behaved model with 2016 data for ease of comparison and defer the problem of data separation until later. Our predictor variables were the percentage of international students enrolled in the program ($x_1$), the percentage of international members on university's board ($x_2$), and the percentage of full-time faculty with doctorates ($x_3$).  Results appear in Table \ref{tbl:ft50results}.
For ease of comparison we place the one-augmented beta regression logit model $d$ coefficients in rows that correspond to the $d$ coefficients from \eqref{eq:thetamodel1}.

\begin{table}[h!t]
\caption{Percent employed after graduation, 2016: Maximum likelihood parameter estimates, standard errors and $z$ statistics.}
\label{tbl:ft50results}
\begin{center}
\tabcolsep=0.10cm
\begin{tabular}{l*{5}{c}}
            &\multicolumn{1}{c}{Augmented}&\multicolumn{1}{c}{}&\multicolumn{1}{c}{}&\multicolumn{1}{c}{}&\multicolumn{1}{c}{Model 2}\\
Para-       &\multicolumn{1}{c}{B.R.}&\multicolumn{1}{c}{Model 1}&\multicolumn{1}{c}{Model 2}&\multicolumn{1}{c}{Model 3}&\multicolumn{1}{c}{with \eqref{eq:thetamodel1}}\\
meter       & ($y_\beta$ \& $z_1$) & ($y$) & ($y$) & ($y$)     & ($y$)      \\
\hline
$b_0$       &      2.2994&      2.3609&      2.3156&      2.3156&      2.3156\\
            &    (0.6444)&    (0.6046)&    (0.6496)&    (0.6496)&    (0.6496)\\
            &    \itshape   3.568&  \itshape     3.905&    \itshape   3.565&     \itshape  3.565&    \itshape   3.565\\
$b_1$       &     -0.0061&     -0.0097&     -0.0069&     -0.0069&     -0.0069\\
            &    (0.0025)&    (0.0027)&    (0.0026)&    (0.0026)&    (0.0026)\\
            &      \ii -2.415&      \ii -3.597&      \ii -2.701&      \ii -2.701&      \ii -2.701\\
$b_2$       &     -0.0004&      0.0021&     -0.0001&     -0.0001&     -0.0001\\
            &    (0.0028)&    (0.0027)&    (0.0028)&    (0.0028)&    (0.0028)\\
            &      \ii -0.152&       \ii 0.775&      \ii -0.037&      \ii -0.037&      \ii -0.037\\
$b_3$       &      0.0009&      0.0026&      0.0013&      0.0013&      0.0013\\
            &    (0.0073)&    (0.0069)&    (0.0073)&    (0.0073)&    (0.0073))\\
            &       \ii 0.131&       \ii 0.370&       \ii 0.184&       \ii 0.184&      \ii  0.184\\
$d_0$       &     -2.2286&            &            &            &	2.2286\\
            &    (4.6172)&            &            &            &(4.6172)\\
            &      \ii -0.483&            &            &            &\ii 0.483\\
$d_1$       &     -0.0507&            &            &            &0.0507\\
            &    (0.0262)&            &            &            &(0.0262)\\
            &      \ii -1.937&            &            &            &\ii 1.937\\
$d_2$       &      0.0154&            &            &            &-0.0154\\
            &    (0.0224)&            &            &            &(0.0224)\\
            &       \ii 0.685&            &            &            &\ii -0.685\\
$d_3$       &      0.0071&            &            &            &-0.0071\\
            &    (0.0503)&            &            &            &(0.0503)\\
            &       \ii 0.142&            &            &            &\ii -0.142\\
$\log(\phi)$&    3.1185  &      3.9916&      3.1325&      3.1325&      3.1325\\
            &    (0.1446)&    (0.3476)&    (0.1428)&    (0.1428)&    (0.1428)\\
            &    \ii 21.568  &      \ii 11.484&      \ii 21.937&      \ii 21.937&      \ii 21.937\\
logit$(\theta)$ &        &      0.3657&      3.1884&            &            \\
            &            &    (0.4697)&    (0.5102)&            &            \\
            &            &       \ii 0.779&       \ii 6.249&            &            \\
\hline
AIC         &   -224.88&   -270.72&   -243.64&   -279.31&   -242.62\\
\hline
\end{tabular}
\caption*{\footnotesize Model outcome variable(s), $y_\beta$ \& $z_1$ or $y$, listed below model name. }
\end{center}
\end{table}

We see that $b_1$ is significant at the $p < 0.05$ level in augmented beta regression while $d_1$ is marginally significant.
The $b_1$ coefficient in Model 1 is also significant and the coefficient is further from zero.  Model 1 has a larger value of $\log(\phi)$ (3.9916 vs. 3.1185)  but is less precisely estimated.  Inverting the logit transformation for Model 1's $\theta$ gives a value of 59.04\%.  This suggests that the tilting power component applies to a sizable portion of the observations.  Model 2 provides similar findings for $b_1$ but $\theta$ is 96.04\% (logit$(\theta) = 3.1884$) since the data set has four observations of 1 out of $N=101$. 
Model 3 omits $\theta$ entirely.  We see the Model 3 coefficients are identical to those of Model 2 excepting the omitted $\theta$ parameter.  The final column of Table \ref{tbl:ft50results} contains results for Model 2 when \eqref{eq:thetamodel1} is estimated.  We see that the value of $d_0$ through $d_3$ replicate those obtained by the one-augmented beta regression model with the exception of a sign change.

\subsection{Model Diagnostics}
One benefit of the proposed Models is many standard model diagnostic techniques will work for the entirety of $y$ with minimal modification since $y$ does not need a tripartite conceptualization.  For example, Figure \ref{fig:scatter2} plots predicted versus observed values for $y\in [0,1]$.  This allows simultaneous examination of model fit properties for endpoint and non-endpoint values.  A residual plot would also be straightforward since we need not concern ourselves with $z_0$ and/or $z_1$.
We can also create a scatter plot of $g(\hat{\mu}_y)$ and $g(\hat{\theta})$ when using Model 2 with \eqref{eq:thetamodel1} (Figure \ref{fig:scatter3}). These plots would be less straightforward to create in augmented beta regression.
Model 3 in particular should inherit many classic beta regression model diagnostics with minor modifications since we have shown it to be a member of the classic beta regression model family.  For example, \citet[][p.~806]{ferrari04} proposed the standardized residual for classic beta regression as
\begin{equation*}
  r_i = \frac{y_i-\mu_i}{\sqrt{\textrm{Var}(y_i) }}
\end{equation*}
where $\textrm{Var}(y_i) = \mu_i(1-\mu_i)/(\phi+1)$. We may use this expression for Model 3 with the revision that  $\textrm{Var}(y_i) = \mu_i(1-\mu_i)/(\phi^*+1)$.

\begin{figure}[ht]
\caption{Example diagnostic plots: 2016 Percent employed after graduation data}
   \label{fig:1-2}
    \begin{minipage}[t]{.45\textwidth}
        \centering
        \includegraphics[width=\textwidth]{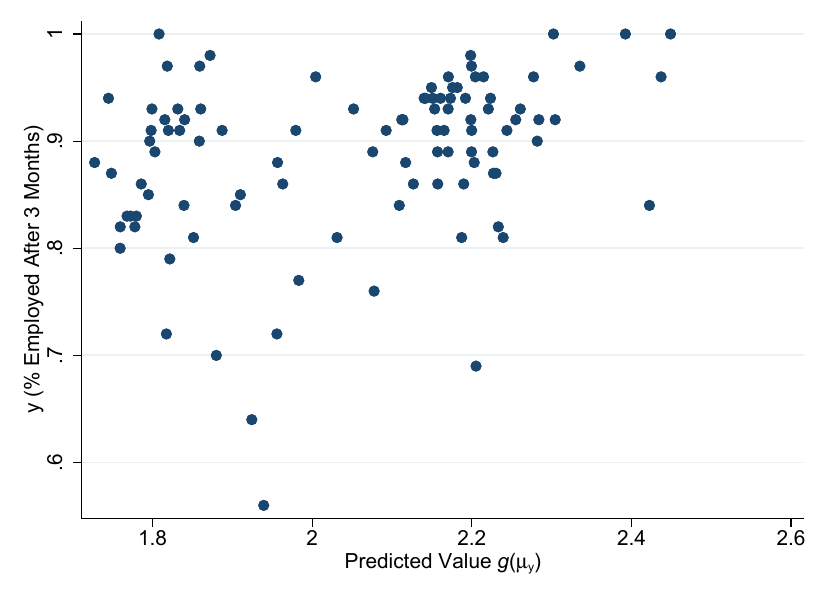}
        \subcaption{Model 3 observed $y$ vs. predicted $g(\mu)$.}\label{fig:scatter2}
    \end{minipage}
    \hfill
    \begin{minipage}[t]{.45\textwidth}
        \centering
        \includegraphics[width=\textwidth]{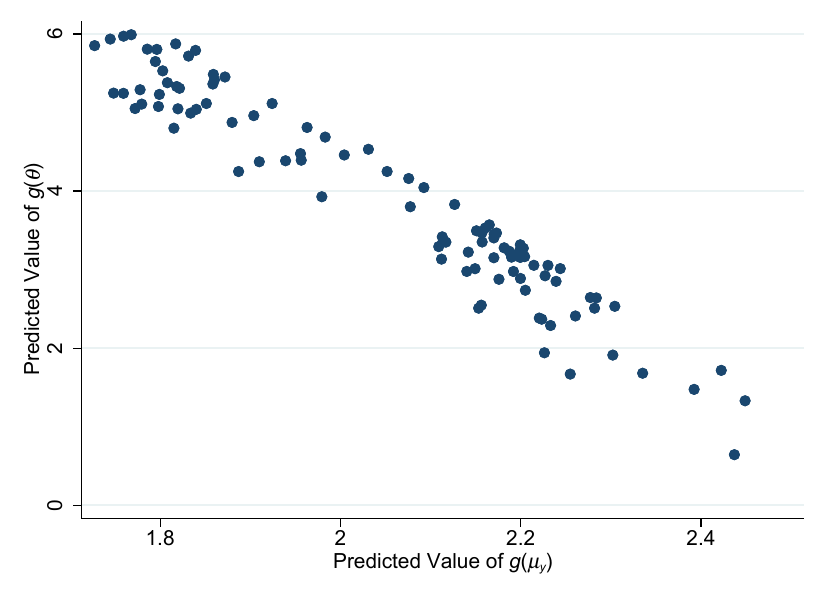}
        \subcaption{Model 2 with \eqref{eq:thetamodel1} }\label{fig:scatter3}
    \end{minipage}
\end{figure}

One application of the models in this paper is to fit data for which augmented beta regression is unsuccessful due to partial or complete data separation \citep{albertanderson84, allison04} in the logit portion of the model. The yearly subsets of our data exhibited partial or complete data separation in all years except one.  Table \ref{tbl:convergence} provides example results for 2014 data.   The $b$ coefficients for the augmented beta regression model appear reasonable but the $d$ ones appear less so.  The $b$ coefficients for Model 3 appear reasonable and also allow us to draw conclusions about $y\in[0,1]$.  Model 2 with \eqref{eq:thetamodel1} omitted produces $b$ estimates that are identical to those of Model 3 as expected (results omitted for economy). Also as expected, Model 2 with \eqref{eq:thetamodel1} produces problematic $d$ estimates just like augmented beta regression (results omitted).  We see this problem is not specific to augmented beta regression but instead is  specific to the binary logit model (which appears in our Model 2 with \eqref{eq:thetamodel1}).  Model 3 omits \eqref{eq:thetamodel1}, avoiding this problem.  \cite{allison04} gives recommendations for addressing data separation.

\begin{table}[ht]
\centering 
\caption{Percent employed after graduation, 2014: Estimates comparison}
\label{tbl:convergence}
\begin{tabular}{l*{2}{cc}}
                    &\multicolumn{2}{c}{Augmented}  &\multicolumn{2}{c}{Model 3}  \\
                    &\multicolumn{2}{c}{B.R.}&\multicolumn{2}{c}{}\\

Para-               & MLE       &             &  MLE            &         \\
meter               & (S.E.)  &  $z$   &  (S.E.)  &  $z$    \\
\hline
$b_0$               &       1.934&     \textit{2.29}&       1.982&    \textit{2.35}\\
                    & (0.8433) &         &  (0.8453)\\
$b_1$               &    -0.0061&     \textit{-2.26}&    -0.0065&     \textit{-2.40}\\
                    & (0.0027)&             & (0.0027)&          \\
$b_2$               &    0.0003&      \textit{0.12}&    0.0006&      \textit{0.20}\\
                    & (0.0027)&             & (0.0027)&          \\
$b_3$               &     0.0020&     \textit{0.22}&     0.0017&     \textit{0.18}\\
                    & (0.0089)&             & (0.0090)&          \\
$d_0$   &      1105.9&      \textit{0.02}&            &            \\
                    & (48627.2)&             &  &          \\
$d_1$   &      -8.648&     \textit{-0.02}&            &            \\
                    & (385.5)&             &  &          \\
$d_2$   &       2.422&      \textit{0.02}&            &            \\
                    & (116.3)&             &  &          \\
$d_3$   &      -12.46&     \textit{-0.02}&            &            \\
                    & (548.3)&             &  &          \\
$\log(\phi)$        &       2.805&     \textit{19.72}&       2.809&     \textit{19.82}\\ \hline
AIC        &    -213.82 &            &         -226.00 &            \\
\hline
\end{tabular}
\end{table}

\subsection{Analysis of Rescaled Healthcare Data\label{sec:gheorge}}
Panel data is a context where augmented beta regression can become challenging.  The models proposed here allow us to examine observation-level dependence in $y_{it}$ where $t$ indexes time, as opposed to nudging values or to modeling $\{y_{\beta i t}, z_{0it}, z_{1it}\}$ which additionally have intrinsic non-temporal correlation by construction. \citet{gheorgeetal17} examined the decline in health-related quality of life (HRQOL) that occurs as people age.  They hypothesized that it is not aging but instead time to death (TTD) that is a more important predictor of health declines in the aged. They used Bayesian mixed beta regression to examine this in a sample drawn from the Netherlands.  Their data was $y \in (0,1]$ where $y=1$ indicates full health as measured by responses to the SF-6D health questionnaire.  \citet{gheorgeetal17} described nudging endpoint values of 1 to 0.99 per \citet{smithson06}.  Here we apply Model 2 to Gheorge et al.'s panel data using portions of their code (available from their publication).


\begin{table}[ht]
\caption{Health-related quality of life: Posterior means and 95\% credible intervals (in parentheses)}
\label{tbl:panelgheorge}
\begin{center}
\footnotesize										
\begin{tabular}{l*{3}{c}}
	       &                &	Beta Regression		     &		Model 2		\\
Variable    & Parameter   &	Rescaled $y\in(0,0.99]$    &	$y\in[0,1]$			\\
\hline								
    &   $b_0$	&	1.282		&		1.236		\\
	&          &	(1.210, 1.355)		&		(1.172, 1.302)		\\
Age & $b_1$  	&	-0.0036		&		-0.0027		\\
    &           &	(-0.0140, 0.0067)		&		(-0.0114, 0.0062)		\\
Gender& $b_2$	&	-0.2179		&		-0.2014		\\
    &           &	(-0.3627, -0.0705)		&		(-0.3293, -0.0731)		\\
TTD	& $b_3$	&	0.0017		&		0.0016		\\
    &           &	(0.0005, 0.0028)		&		(0.0006, 0.0027)		\\
    & $d_0$	    &			&		4.240		\\
	&          &			&		(3.507, 5.066)		\\
Age &$d_1$	&			&		-0.0151		\\
    &			&		     &       (-0.1120, 0.0799)		\\
Gender& $d_2$	&			&		-0.2017		\\
    &           &			&		(-0.3372, -0.0675)		\\
TTD & $d_3$	&			&		0.0001		\\
    &           &		&		(-0.0132, 0.0135)		\\
    &   $\phi$	&	25.89		&		25.91		\\
	&          &	(21.08, 31.20)		&		(20.34, 31.98)		\\
\hline
\end{tabular}
\end{center}
\end{table}
We restrict ourselves to 556 responses data on 356 individuals where missingness did not occur.  Ten observations where HRQOL was 0.99 were transformed back to 1.  We used the functional form of their Table 4 Model 6 which has individual-specific random effects and their priors with one exception. We used $\phi \sim \textrm{Uniform}(3,200)$ while their prior was on the reciprocal of $\phi$.    Estimation was based on two chains started from different initial values.  Each chain had 5000 iterations of burn-in, and a subsequent 200,000 iterations were retained.  The two chains converged to a common distribution for all parameters.  All ESS's for the $b$ parameters were greater than 4500.  All ESS's for the $d$ parameters were greater than 1400. Results appear in Table \ref{tbl:panelgheorge}.
We include descriptors (age, gender, and TTD) to promote comparison with  \citet{gheorgeetal17}.
Table \ref{tbl:panelgheorge} shows that the $b$ coefficients from classic beta regression using rescaled data are similar to those of Model 2.  Model 2's $b$ coefficients are slightly smaller and also have narrower 95\% credible intervals.  Both models indicate males have lower HRQOL than females ($b_2$), and that greater TTD predicts better HRQOL ($b_3$).   Model 2 also produces $d$ coefficients for insights about endpoint observations, extending their findings. Per $d_2$, females are more likely to report perfect health ($y=1$) than are males.

\section{Concluding Remarks\label{sec:conc}}
It is commonly believed that the beta distribution cannot be used for $y\in[0,1]$ data \citep[][p.~124]{zhouandhuang22,ospina10}.
This paper expands the palette of models for $y\in[0,1]$ by providing an endpoint-heterogeneous beta regression model for $y\in[0,1]$ (Model 3), three additional closely-related beta regression models, and allowing the first unconstrained modeling of the marginal expectation without data rescaling or data nudging.
Despite our contrasting modeling approach and likelihood functions, in certain cases we can exactly replicate results obtained by augmented beta regression (up to a possible coefficient sign change which does not affect coefficient $p$ values or model comparison measures).  Hence in addition to our contributions of new models for $y\in[0,1]$ data, we provide a new lens for viewing augmented beta regression and its binary logit formulation.  The models can have varying amounts of parsimony in that we can assign both endpoints to the tilting power component in Model 2 (instead of treating them separately) or by treating all observations in a unified manner (Model 3) using endpoint-heterogeneity.  This parsimony may be useful for interpretation and advantageous in more complex modeling situations such as NMAR data, difference-in-differences estimation  or spatiotemporal applications.  One of the limitations of augmented beta regression is that the logit portion may be inestimable due to data separation as in Table \ref{tbl:convergence}.  This may be more likely when there are only a few endpoint values since it may be easier to perfectly predict the few endpoint values.  The models of this paper do not suffer from this problem as long as \eqref{eq:thetamodel1} is omitted. Furthermore the simulation shows that the current models' performance is particularly good when $|n_0-n_1|$ is small.  Finally, the current models remove the need for sensitivity analyses that are associated with data rescaling.  \citet{gheorgeetal17} described performing sensitivity analyses for various choices of the alteration of the endpoint value when using the data re-analyzed in Section \ref{sec:gheorge}.  This step was not necessary for the current paper's models in Table \ref{tbl:panelgheorge}. As for limitations, a unique MLE may not always exist, particularly when only one endpoint predominates in that parameter's $y$.  The existence can be readily checked for group means and intercepts as shown above.  For slope parameters a simple check remains an area of research and options include omitting the parameter, using Model 4, or possibly even selectively-performed minor amounts of data rescaling.  Finally, we identify conditions where augmented beta regression will likely be preferable to the models presented in this paper (Remark 1) and illustrate that the strategies of data rescaling or data nudging can now be de-emphasized.  Thus the paper provides both new models as well as new strategies for modeling practice in $y\in[0,1]$ data.

\bibliographystyle{ejor06}
\bibliography{betapalettearxivV2}

\newcommand{\noop}[1]{}
\begin{thebibliography}{38}
\providecommand{\natexlab}[1]{#1}
\expandafter\ifx\csname urlstyle\endcsname\relax
  \providecommand{\doi}[1]{doi:\discretionary{}{}{}#1}\else
  \providecommand{\doi}{doi:\discretionary{}{}{}\begingroup
  \urlstyle{rm}\Url}\fi

\bibitem[{Abramowitz \protect\BIBand{} Stegun(1964)}]{abramowitz64}
Abramowitz M, Stegun I.
\newblock Handbook of Mathematical Functions with Formulas, Graphs, and
  Mathematical Tables.
\newblock U.S. Govt. Printing Office, Washington, DC, 1964.

\bibitem[{Aitkin \protect\BIBand{} Wilson(1980)}]{aitkinwilson80}
Aitkin M, Wilson GT.
\newblock Mixture models, outliers, and the {EM} algorithm.
\newblock Technometrics 1980;\hspace{0pt}22(3); 325--331.

\bibitem[{Albert \protect\BIBand{} Anderson(1984)}]{albertanderson84}
Albert A, Anderson JA.
\newblock On the existence of maximum likelihood estimates in logistic
  regression models.
\newblock Biometrika 1984;\hspace{0pt}71(1); 1--10.

\bibitem[{Allison(2004)}]{allison04}
Allison P.
\newblock Convergence problems in logistic regression.
\newblock In: Altman M, Gill J, McDonald MP (Eds.), Numerical Issues in
  Statistical Computing for the Social Scientist, Wiley, Hoboken, NJ. pp.
  238--252.

\bibitem[{Bayes et~al.(2012)Bayes, B\'{a}zan, \protect\BIBand{}
  Garc\'{i}a}]{bayesetal12}
Bayes CL, B\'{a}zan JL, Garc\'{i}a C.
\newblock A new robust regression model for proportions.
\newblock Bayesian Analysis 2012;\hspace{0pt}7(4); 841--866.

\bibitem[{Buis(2010)}]{buis12}
Buis ML.
\newblock {ZOIB}: Stata module to fit a zero-one inflated beta distribution by
  maximum likelihood.
\newblock Statistical Software Components S457156, Boston College Department of
  Economics, revised 08 Aug 2012, 2010.

\bibitem[{Canterle \protect\BIBand{} Bayer(2019)}]{canterlebayer19}
Canterle DR, Bayer FM.
\newblock Variable dispersion beta regressions with parametric link functions.
\newblock Statistical Papers 2019;\hspace{0pt}60; 1541--1567.

\bibitem[{Cook et~al.(2008)Cook, Kieschnick, \protect\BIBand{}
  McCullough}]{cooketal08}
Cook DO, Kieschnick R, McCullough B.
\newblock Regression analysis of proportions in finance with self selection.
\newblock Journal of Empirical Finance 2008;\hspace{0pt}15(5); 860--867.

\bibitem[{Da-{S}ilva \protect\BIBand{} Migon(2016)}]{dasilva16}
Da-{S}ilva CQ, Migon HS.
\newblock Hierarchical dynamic beta model.
\newblock Revstat-Statistical Journal 2016;\hspace{0pt}14(1); 49--73.

\bibitem[{Dezs\H{o} et~al.(2022)Dezs\H{o}, Alm, \protect\BIBand{}
  Kirchler}]{dezsoetal21}
Dezs\H{o} L, Alm J, Kirchler E.
\newblock Inequitable wages and tax evasion.
\newblock Journal of Behavioral and Experimental Economics 2022;\hspace{0pt}96;
  101811.

\bibitem[{Di~Brisco \protect\BIBand{}
  Migliorati(2020)}]{dibriscoandmigliorati20}
Di~Brisco AM, Migliorati S.
\newblock A new mixed-effects mixture model for constrained longitudinal data.
\newblock Statistics in Medicine 2020;\hspace{0pt}39(2); 129--145.

\bibitem[{Ferrari \protect\BIBand{} Cribari-{N}eto(2004)}]{ferrari04}
Ferrari SLP, Cribari-{N}eto F.
\newblock Beta regression for modelling rates and proportions.
\newblock Journal of Applied Statistics 2004;\hspace{0pt}31(7); 799--815.

\bibitem[{{Financial~Times}(2022)}]{FT22}
{Financial~Times}.
\newblock Global {MBA} ranking.
\newblock Retrieved May 7, 2022 from
  \url{https://rankings.ft.com/home/masters-in-business-administration}, 2022.

\bibitem[{Galivs et~al.(2014)Galivs, Bandyopadhyay, \protect\BIBand{}
  Lachos}]{galivs14}
Galivs D, Bandyopadhyay D, Lachos V.
\newblock Augmented mixed beta regression models for periodontal proportion
  data.
\newblock Statistics in Medicine 2014;\hspace{0pt}33(21); 3759--–3771.

\bibitem[{Gheorghe et~al.(2017)Gheorghe, Picavet, Verschuren, Brouwer,
  \protect\BIBand{} van Baal}]{gheorgeetal17}
Gheorghe M, Picavet S, Verschuren M, Brouwer WBF, van Baal PHM.
\newblock Health losses at the end of life: a {B}ayesian mixed beta regression
  approach.
\newblock Journal of the Royal Statistical Society. Series A (Statistics in
  Society) 2017;\hspace{0pt}180(3); 723--749.

\bibitem[{Guolo \protect\BIBand{} Varin(2014)}]{guolo14}
Guolo A, Varin C.
\newblock Beta regression for time series analysis of bounded data, with
  application to {C}anada {G}oogle flu trends.
\newblock The Annals of Applied Statistics 2014;\hspace{0pt}8(1); 74--88.

\bibitem[{Hahn(2021)}]{hahn21}
Hahn ED.
\newblock Regression modeling with the tilted beta distribution: A {B}ayesian
  approach.
\newblock The Canadian Journal of Statistics 2021;\hspace{0pt}49(2); 262--282.

\bibitem[{Hunger et~al.(2012)Hunger, D{\"o}ring, \protect\BIBand{}
  Holle}]{hunger12}
Hunger M, D{\"o}ring A, Holle R.
\newblock Longitudinal beta regression models for analyzing health-related
  quality of life scores over time.
\newblock BMC Medical Research Methodology 2012;\hspace{0pt}12; 144.

\bibitem[{Hwang et~al.(2021)Hwang, Chu, \protect\BIBand{} Yu}]{hwangetal21}
Hwang RC, Chu CK, Yu K.
\newblock Predicting the loss given default distribution with the zero-inflated
  censored beta-mixture regression that allows probability masses and
  bimodality.
\newblock Journal of Financial Services Research 2021;\hspace{0pt}59(3);
  143--172.

\bibitem[{Kaufeld et~al.(2014)Kaufeld, Heaton, \protect\BIBand{}
  Sain}]{kaufeldetal14}
Kaufeld KA, Heaton MJ, Sain SR.
\newblock A spatio-temporal model for mountain pine beetle damage.
\newblock Journal of Agricultural, Biological, and Environmental Statistics
  2014;\hspace{0pt}19(4); 439--452.

\bibitem[{Kieschnick \protect\BIBand{} Mc{C}ullough(2003)}]{kies03}
Kieschnick R, Mc{C}ullough B.
\newblock Regression analysis of variates observed on $(0,1)$: Percentages,
  proportions and fractions.
\newblock Statistical Modelling 2003;\hspace{0pt}3; 193--213.

\bibitem[{Mandal et~al.(2016)Mandal, Srivastav, \protect\BIBand{}
  Simonovic}]{mandaletal16}
Mandal S, Srivastav RK, Simonovic SP.
\newblock Use of beta regression for statistical downscaling of precipitation
  in the {C}ampbell {R}iver basin, {B}ritish {C}olumbia, {C}anada.
\newblock Journal of Hydrology 2016;\hspace{0pt}538; 49--62.

\bibitem[{Ospina \protect\BIBand{} Ferrari(2010)}]{ospina10}
Ospina R, Ferrari SLP.
\newblock Inflated beta distributions.
\newblock Statistical Papers 2010;\hspace{0pt}51; 111--126.

\bibitem[{Paolino(2001)}]{paolino01}
Paolino P.
\newblock Maximum likelihood estimation of models with beta-distributed
  dependent variables.
\newblock Political Analysis 2001;\hspace{0pt}9(4); 325--346.

\bibitem[{Ribeiro et~al.(2021)Ribeiro, Nobre, dos Santos, \protect\BIBand{}
  Azevedo}]{Ribeiroetal21}
Ribeiro VSO, Nobre JS, dos Santos JRS, Azevedo CLN.
\newblock Beta rectangular regression models to longitudinal data.
\newblock Brazilian Journal of Probability and Statistics
  2021;\hspace{0pt}35(4); 851--874.

\bibitem[{Rocha \protect\BIBand{} Cribari-Neto(2008)}]{rocha08}
Rocha AV, Cribari-Neto F.
\newblock Beta autoregressive moving average models.
\newblock Test 2008;\hspace{0pt}18(3); 529.

\bibitem[{Smithson \protect\BIBand{} Verkuilen(2006)}]{smithson06}
Smithson M, Verkuilen J.
\newblock A better lemon squeezer? {M}aximum-likelihood regression with
  beta-distributed dependent variables.
\newblock Psychological Methods 2006;\hspace{0pt}11(1); 54--71.

\bibitem[{Souza \protect\BIBand{} Moura(2016)}]{SouzaMoura16}
Souza DF, Moura FAS.
\newblock Multivariate beta regression with application in small area
  estimation.
\newblock Journal of Official Statistics 2016;\hspace{0pt}32(3); 747--768.

\bibitem[{Titterington(2011)}]{titterington11}
Titterington DM.
\newblock The {EM} algorithm, variational approximations and expectation
  propagation for mixtures.
\newblock In: Mengerson KL, Robert CP, Titterington DM (Eds.), Mixtures:
  Estimation and Applications, Wiley, Chichester, UK. pp. 213--239.

\bibitem[{Tobin(1958)}]{tobin58}
Tobin J.
\newblock Estimation of relationships for limited dependent variables.
\newblock Econometrica 1958;\hspace{0pt}26(1); 24--36.

\bibitem[{Tran et~al.(2022)Tran, Vu, Ngo, Tran, \protect\BIBand{}
  Ho}]{tranetal22}
Tran PT, Vu BT, Ngo ST, Tran VD, Ho TD.
\newblock Climate change and livelihood vulnerability of the rice farmers in
  the {N}orth {C}entral region of {V}ietnam: A case study in {N}ghe {A}n
  province, {V}ietnam.
\newblock Environmental Challenges 2022;\hspace{0pt}7; 100460.

\bibitem[{van Dorp \protect\BIBand{} Kotz(2002)}]{vandorp02b}
van Dorp R, Kotz S.
\newblock The standard two-sided power distribution and its properties: With
  applications in financial engineering.
\newblock The American Statistician 2002;\hspace{0pt}56(2); 90--99.

\bibitem[{Verkuilen \protect\BIBand{} Smithson(2012)}]{Verkuilen12}
Verkuilen J, Smithson M.
\newblock Mixed and mixture regression models for continuous bounded responses
  using the beta distribution.
\newblock Journal of Educational and Behavioral Statistics
  2012;\hspace{0pt}37(1); 82--113.

\bibitem[{Wang \protect\BIBand{} Luo(2016)}]{wangandluo16}
Wang J, Luo S.
\newblock Augmented beta rectangular regression models: A {B}ayesian
  perspective.
\newblock Biometrical Journal 2016;\hspace{0pt}58(1); 206--221.

\bibitem[{Wang \protect\BIBand{} Luo(2017)}]{wangandluo17}
Wang J, Luo S.
\newblock Bayesian multivariate augmented beta rectangular regression models
  for patient-reported outcomes and survival data.
\newblock Statistical Methods in Medical Research 2017;\hspace{0pt}26(4);
  1684--1699.

\bibitem[{Zhao et~al.(2014)Zhao, Zhang, Lv, \protect\BIBand{} Liu}]{zhaoetal14}
Zhao W, Zhang R, Lv Y, Liu J.
\newblock Variable selection for varying dispersion beta regression model.
\newblock Journal of Applied Statistics 2014;\hspace{0pt}41(1); 95--108.

\bibitem[{Zhou \protect\BIBand{} Huang(2022)}]{zhouandhuang22}
Zhou H, Huang X.
\newblock Bayesian beta regression for bounded responses with unknown supports.
\newblock Computational Statistics \& Data Analysis 2022;\hspace{0pt}167;
  107345.

\bibitem[{Zimprich(2010)}]{zimprich10}
Zimprich D.
\newblock Modeling change in skewed variables using mixed beta regression
  models.
\newblock Research in Human Development 2010;\hspace{0pt}7(1); 9--26.

\end{thebibliography}
\appendix
\section*{Appendix}
  Note that \eqref{eq:LLy0} and \eqref{eq:LLy1} in isolation are unlikely to be concave functions, necessitating this proof.  We seek to prove that the likelihood function resulting from combining \eqref{eq:LLy0} with \eqref{eq:LLybeta}, or alternatively \eqref{eq:LLy1} with \eqref{eq:LLybeta}, has an optimum in the range $(0,1)$.  First consider combining \eqref{eq:LLy0} with \eqref{eq:LLybeta}.
  To begin, note that we have a classic beta regression model for $y_\beta$ with global optimum for $(\mu,\phi)$ by assumption.  Accordingly the terms on the second line of \eqref{eq:LLybeta} have a maximum, since the second line of \eqref{eq:LLybeta} is the kernel of \eqref{eq:LLybeta}.  Therefore the existence of an optimum as well as concavity holds for these terms on the second line of \eqref{eq:LLybeta}.  It then remains to consider the first line of \eqref{eq:LLybeta} along with \eqref{eq:LLy0}. 
  Now add $n_0$ observations in which $y=0$ to the data per \eqref{eq:LLy0} and retain only normalizing constants from the first line of \eqref{eq:LLybeta} to produce
  \begin{equation}
  \log \ell_{C_0} = n_0  \bigl( \log(1-\mu) - \log  (\mu ) \bigr)
   +
      n_\beta \bigl( \log {\Gamma (\phi )} -   \log {\Gamma (\mu  \phi )} - \log{ \Gamma (\phi -\mu  \phi )}\bigr).
  \label{eq:normconsts0}
  \end{equation}
  A condition of the proof is $n_\beta \geq N/2$ which implies $n_0 \leq n_\beta$. Upon setting $n_0=n_\beta$ in \eqref{eq:normconsts0} we see these quantities can be temporarily ignored for the purposes of finding an optimum.  Taking the derivative of \eqref{eq:normconsts0} with respect to $\mu$ when $n_0$ and $n_\beta$ are ignored gives
  \begin{equation}\label{eq:normconst0derivmu}
  \frac {\partial\log \ell_{C_0}}{\partial\mu}= \frac{1}{\mu^2-\mu}-\phi\, \psi\bigl( \mu\phi \bigr) +\phi\, \psi\bigl( \phi(1-\mu)\bigr)
  \end{equation}
  where $\psi(\cdot)$ is the digamma function namely $\psi(z) = \textrm{d} \log \Gamma(z)/\textrm{d}z$.  A series expansion of $\psi(\cdot)$  aids in the proof.  As $z \rightarrow \infty$, \cite{abramowitz64} observe that
  \begin{align*}
  \psi(z) & = \log(z) - 1/2z - 1/12z +\cdots\label{eq:digamma}\\
  \psi^\prime(z) & = 1/z + 1/2z^2 + 1/6z^3 + \cdots.
  \end{align*}
  Note that as $\phi \rightarrow \infty$ so do the quantities inside $\psi(\cdot)$ in \eqref{eq:normconst0derivmu}.
  Suppose we truncate the expansion of $\psi(z)$ to the first two terms and accordingly substitute into the righthand side of \eqref{eq:normconst0derivmu}, then set the righthand side equal to zero.  After some manipulation we find that the first-order condition of $\frac {\partial\log \ell_{C_0}}{\partial\mu}=0$ occurs when
  \begin{equation}\label{eq:FOCapprox}
    \frac{2 \mu +1}{2 \left(\mu ^2-\mu \right) \bigl(\log (\mu ) - \log (1-\mu )\bigr)} = \phi.
  \end{equation}
  We can see that for $\phi$ to be large, $\mu$ must approach 1/2 from below, causing $\log (\mu ) - \log (1-\mu )$ to approach zero.
 For the second-order condition we find that
  \begin{equation}\label{eq:normconst2derivmu}
  \frac {\partial^2 \log \ell_{C_0}}{\partial\mu^2}=
 \frac{1}{\mu ^2}-\frac{1}{(1-\mu )^2}
 -\phi ^2 \psi ^\prime(\mu  \phi )-\phi ^2 \psi ^\prime(\phi -\mu  \phi ).
  \end{equation}
We take the first two terms in the expansion of the trigamma $\psi^\prime(\cdot)$.
Inserting \eqref{eq:FOCapprox} into \eqref{eq:normconst2derivmu} we obtain
  \begin{equation}\label{eq:SOCapprox}
   \frac{1 -2 \mu(\mu+1) -\dfrac{2 \mu +1}{2 \log (1-\mu )-2 \log (\mu )}}{(\mu -1)^2 \mu ^2}.
  \end{equation}
  We see the lower denominator is always positive and can be ignored for the purposes of determining the sign.
  Considering the numerator further we see that $1-2\mu(\mu+1)$ tends toward -1/2 as $\mu$ tends to 1/2 from below, while $2\mu+1$ and $2\log(1 - \mu) - 2\log(\mu)$ tend toward 2 and zero respectively as $\mu$ tends to 1/2. Accounting for the minus sign, we confirm the Hessian is negative indicating a maximum and therefore \eqref{eq:normconsts0} is concave under the conditions of the proof.

  To summarize, the second line of \eqref{eq:LLybeta} is concave by assumption and we have shown that the sum of \eqref{eq:LLy0} and of the first line of \eqref{eq:LLybeta} is concave when $n_0 = n_\beta$.  Since a sum of two concave functions is concave, the proof is complete for $n_0 = n_\beta$.  Now consider whether the sum of \eqref{eq:LLy0} and of the first line of \eqref{eq:LLybeta} is concave when $n_0 < n_\beta$.  Applying again the concept of the sum of two concave functions being concave, note that if the sum of \eqref{eq:LLy0} and of the first line of \eqref{eq:LLybeta} is concave when $n_0 = n_\beta$, then this sum is also concave when $n_0 < n_\beta$, as in this instance then contribution of \eqref{eq:LLy0} is smaller than when $n_0=n_\beta$.  This completes the proof for the cases where $n_0 < n_\beta$.

  Next we examine the case where the likelihood function results from adding \eqref{eq:LLy1} to \eqref{eq:LLybeta} and where $n_1=n_\beta$.  It is possible to repeat the above process; however for economy note that this case is a reflection around the midpoint of a case we have already proved.  Namely if we reflect the data by transforming to $1-y$, we may apply the proof above to this transformed variable. Accordingly this proves the existence of an optimum when $n_1 = n_\beta$.  By extension of an above case already proved, this also proves the case when $n_1 < n_\beta$.
\end{document}